\documentclass[preprint]{elsarticle}

\usepackage{amsmath}
\usepackage{amssymb}
\usepackage{multicol}
\usepackage{color}
\usepackage{xspace}
\usepackage{pdfwidgets}

\newtheorem{Def}{Definition}
\newtheorem{Lem}{Lemma}
\newtheorem{Theo}{Theorem}
\newtheorem{Cor}{Corollary}
\newproof{proof}{Proof}

\begin{document}

\title{Sequentialization for full N-Graphs via sub-N-Graphs}

\begin{frontmatter}

\author[rural,federal]{Ruan V. B. Carvalho}
\ead{ruan.carvalho@ufrpe.br, rvbc@cin.ufpe.br}

\author[federal]{La\'{i}s. S. Andrade}
\ead{lsa@cin.ufpe.br}

\author[federal]{Anjolina G. de Oliveira}
\ead{ago@cin.ufpe.br}

\author[federal]{Ruy. J. G. B. de Queiroz}
\ead{ruy@cin.ufpe.br}

\address[rural]{Depto. de Estat\'{i}stica e Inform\'{a}tica, Universidade Federal Rural de Pernambuco
52171-900 Recife, Pernambuco, Brazil}
\address[federal]{Centro de Inform\'{a}tica, Universidade Federal de Pernambuco,
50740-560 Recife, Pernambuco, Brazil}
\begin{abstract}
Since proof-nets for MLL$^-$ were introduced by Girard (1987), several studies have appeared dealing with its soundness proof.
Bellin \& Van de Wiele (1995) produced an elegant proof based on properties of subnets (empires and kingdoms) and
Robinson (2003) proposed a straightforward generalization of this presentation for proof-nets from sequent
calculus for classical logic.
In 2014 it was presented an extension of these studies to obtain a proof of the sequen\-ti\-a\-li\-za\-tion theorem
for the fragment of N-Graphs with $\wedge, \vee ~ and ~\neg{}$, via the notion of sub-N-Graphs.
N-Graphs is a symmetric natural deduction calculus with multiple conclusions that adopts Danos--Regnier's criterion and has defocussing switchable links.
In this paper, we present a sequentization for full propositional classical N-Graphs, showing how to find a split node in the middle of
the proof even with a global rule for discharging hypothesis.
\end{abstract}


\end{frontmatter}

\section{Introduction}\label{sec:intro}

Since proof-nets for MLL$^-$ were introduced by Girard \cite{girard:1987}, several studies have been made on its soundness proof.
The first correctness criterion defined for proof-nets was given with the definition of the \emph{no shorttrip condition}:
Girard used trips to define empires and proved that if all terminal formulas in a proof-net {\small $R$} 
are conclusions of times links, then there is at least one terminal formula which splits {\small $R$}.
After Danos--Regnier's work \cite{danos:1989} it has become possible to define empires using their newly defined DR graphs and, 
with this new notion of empires, Girard proved sequentialization for proof-nets with quantifiers \cite{girard:1991}.
Another important advance was achieved by the introduction of a new type of subnets, namely kingdoms.
Once the notion of kingdoms was introduced, Bellin \& Van de Wiele produced an elegant proof of the sequentialization theorem using
simple properties of subnets \cite{bellin:1995}.

A straightforward generalization of this proof was obtained by Robinson \cite{robinson:2003}. He pointed out that
Danos--Regnier's technique relies only on the format of the rules and does not depend on the logic involved.
So he devised a proof system based on the classical sequent calculus and applied the characterization of subnets and the proof of
sequentialization for MLL$^-$ to his proof-nets for classical logic. His proof followed the model defined by 
Bellin \& Van de Wiele \cite{bellin:1995}.

However, this generalization does not cover the existence of so called switchable links with one premise and more than one conclusion, and also
the absence of axiom links. In such systems subnets are not necessarily closed under hereditary premises.
So, if a subnet contains a formula occurrence {\small $A$} and {\small $B$} is 
above\footnote{We say that {\small $B$} is above {\small $A$} when {\small $B$} is a hereditary premise of {\small $A$}.} {\small $A$} 
in the proof-net, then {\small $B$} may not be in this subnet.
Other works in linear logic related to these issues are Lafont's interaction nets (which do not have axiom links) \cite{lafont:1995} and
the system of Blute \emph{et al}, which contains such switchable links \cite{blute:1996} and inspired the proof-nets
for classical logic proposed by F\"{u}hrman \& Pym \cite{fuhrmann:2007}. 
Hughes also proposed a graphical proof system for classical logic where proofs are combinatorial rather than syntactic: 
a proof of {\small $A$} is a homomorphism between a coloured graph and a graph associated with {\small $A$} \cite{hughes:2006}.
McKinley, on the other hand, proposed the expansion nets, a system that focus on canonical representation of cut-free proofs \cite{mckinley:2010}.

Here we present an extended version of a previous work \cite{carvalho:2014}
to perform the sequentialization for N-Graphs, 
a multiple conclusion calculi inspired by the proof-nets for the propositional classical logic developed by 
de Oliveira \cite{oliveira:2001:phd,oliveira:2001}, but with a switchable defocussing link and without axiom links.
One of the main results of this paper, besides giving a new soundness proof for N-Graphs, is the definition of a ge\-ne\-ra\-li\-zed method to
make surgical cuts in proofs for classical logic. This comes with the fact that the presence of the split node in an 
N-Graph can occur essentially anywhere in the proof, unlike proof-nets where the split node is always a terminal formula.

The need to identify the split node is at the heart of our proof of the sequentialization. In order to achieve that we define the north, 
the south and the whole empires of a formula occurrence {\small $A$}.
The first one corresponds to the empires notion of Girard's and Robinson's proof-nets. The second one is the largest sub-N-Graph which has 
{\small $A$} as a premise (defined due to the presence of elimination rules in N-Graphs). The last one is the union of the previously 
defined and it induces a strict ordering over the graph nodes, which will be fundamental to find the split node.

The deduction system defined by N-Graphs has a improper inference rule for the introduction of ``{\small $\rightarrow$}'':
it allows discharging hypotheses.
In this paper we also include this connective and show how to find a split node in the middle of the proof even in the
presence of a global rule for the introduction of ``{\small $\rightarrow$}''.


\section{N-Graphs}

Proposed by de Oliveira \cite{oliveira:2001:phd,oliveira:2001}, N-Graphs is a symmetric natural deduction (ND) calculus with the presence of
structural rules, similar to the sequent calculus. It is a multiple conclusion proof system for classical logic where proofs are built in the
form of directed graphs (\emph{``digraphs''}). Several studies have been developed on N-Graphs since its first publication in 2001 
\cite{oliveira:2001:phd}, like Alves' development on the geometric perspective and cycle treatment towards the normalization of 
the system \cite{alves:2005} and Cruz's definition of intuitionistic N-Graphs \cite{cruz:2013}. A normalization algorithm was presented for
classical N-Graphs \cite{alves:2011}, along with the subformula and separation properties \cite{alves:2009}. More recently a linear time
proof checking algorithm was proposed \cite{andrade:2013}. N-Graphs also inspired a natural deduction system for the logic of lattices \cite{restall:2005}.

\subsection{Proof-Graphs}

The system is defined somewhat like the proof-nets. There is the concept of \emph{proof-graphs}, which are all graphs constructed with the 
valid links where each node is the premise and conclusion of \emph{at most} one link, and the concept of \emph{N-Graphs}, which are the 
correct proof-graphs, i.e. the proof-graphs that represent valid proofs. These constructions are analogous to the definition of proof-structure and
proof-net, respectively.

The links represent atomic steps in a derivation. \emph{Focussing links} are the ones with two premises and one conclusion,
as illustrated by Fig.~\ref{fi:focussing-links} ({\small $\wedge-I$}, {\small $\bot-link$}, {\small $\rightarrow-E$},
{\small $\top-focussing~weak$} and contraction).
The \emph{defocussing links} are the ones with one premise and two conclusions, as shown in 
Fig.~\ref{fi:focussing-links} ({\small $\vee-E$}, {\small $\top-link$}, {\small $\rightarrow-I$},
{\small $\bot-defocussing~weak$} and expansion). All other links
are called \emph{simple links} and have only one premise and one conclusion (Fig.~\ref{fi:simple-links}).

There are two kinds of edges (\emph{``solid''} and \emph{``meta''}) and the second one are labeled with an
\emph{``m''} ($(u,v)^m$). The \emph{solid} indegree (outdegree) of a vertex $v$ is the number of solid edges oriented
towards (away from) it. The \emph{meta} indegree and outdegree are defined analogously.
The set of vertices with indegree (outdegree) equal to zero is the set of premises (conclusions) of the proof-graph
{\small $G$}, and is represented by {\small $PREMIS(G)$} ({\small $CONC(G)$}).
The set of vertices with solid indegree equal to zero and meta indegree equal to one is the set of canceled hypothesis of
$G$ ({\small $HYPOT(G)$}).

\begin{figure}[ht]
\centering
\includegraphics[width=240pt]{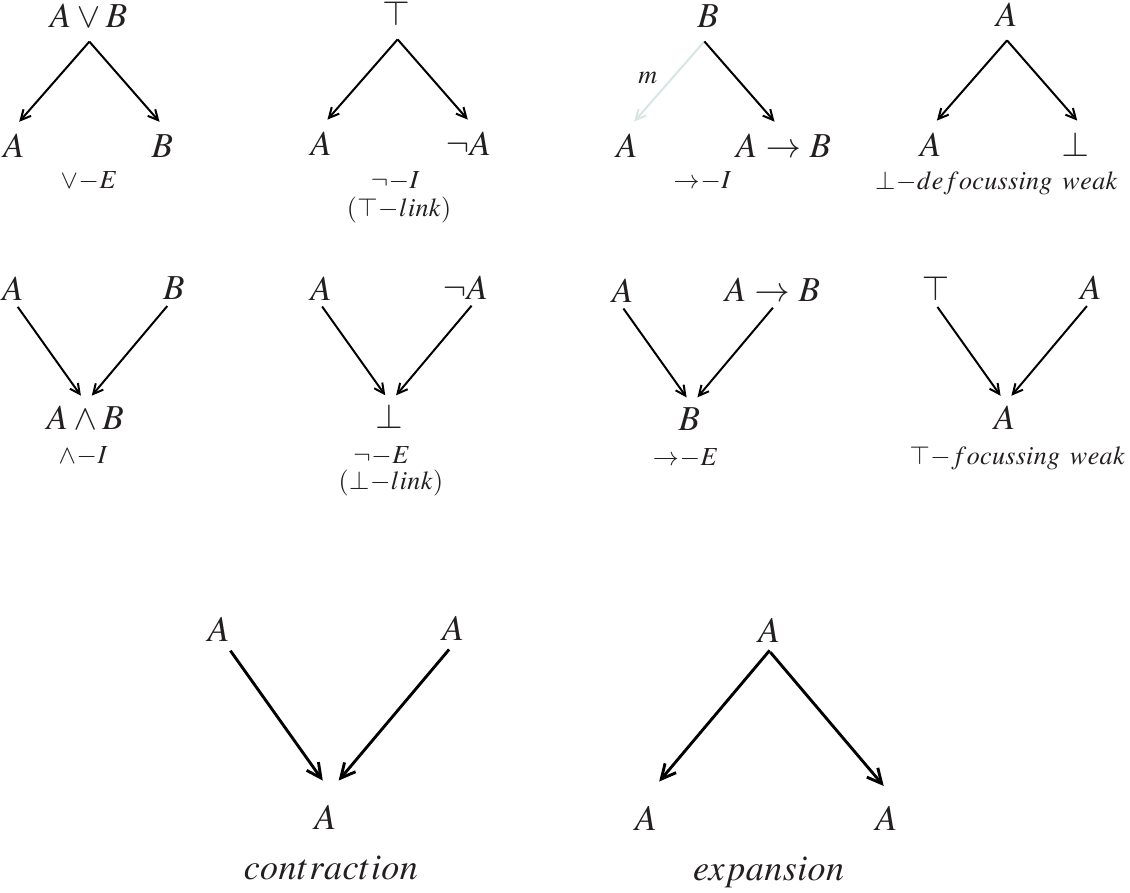}
\caption{Focussing and defocussing links.}
\label{fi:focussing-links}
\end{figure}

\begin{figure}[ht]
\centering
\includegraphics[width=160pt]{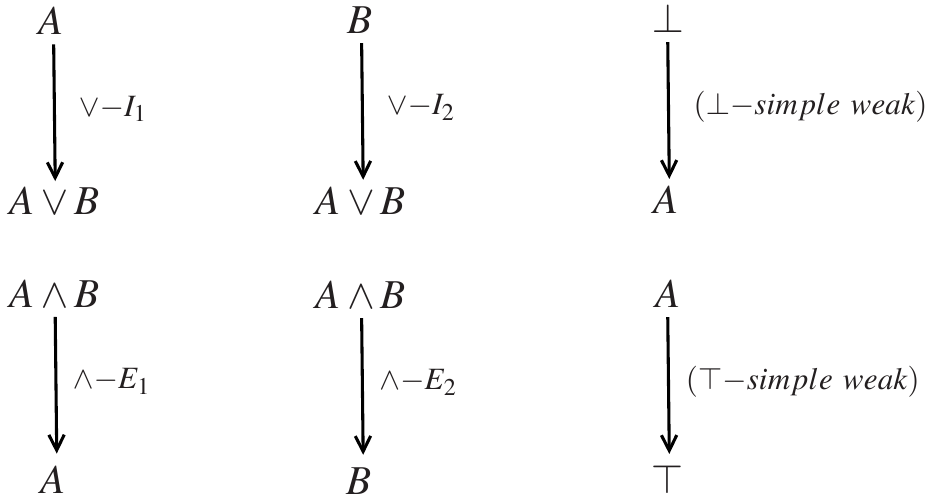}
\caption{Simple links.}
\label{fi:simple-links}
\end{figure}

A \emph{logical link} represents a derivation in ND and we have introduction and elimination links for each connective; {\small $\neg{-I}$} ({\small $\top-link$}) acts as the law of the excluded middle.
A \emph{structural link} expresses the application of a structural rule as it is done in sequent calculus: it enables 
weakening a proof  ({\small $\top-focussing~weak$}, {\small $\bot-defocussing~weak$}, {\small $\top-simple~weak$} and 
{\small $\bot-simple~weak$}), duplicating premises (expansion link) and grouping conclusions in equivalence classes
(contraction link). There is no link to emulate the interchange rule because in a proof-graph the order of 
the premises is not important for the application of derivation rules.

The axioms are represented by proof-graphs with one vertex and no edges. Then, a single node labeled by {\small $A$} is already a valid
derivation: it represents an axiom in sequent calculus ({\small $A \vdash A$}). So here it makes no sense to talk about
the smallest subgraph having {\small $A$} as a conclusion: it would be trivially the vertex {\small $v$} labeled by {\small $A$}.
Therefore the notion of kingdoms, as defined and used by Bellin \& Van de Wiele \cite{bellin:1995} for their sequentialization, is useless for
N-Graphs.

In Fig.~\ref{fi:cycle-examples} there are three proof-graphs. The first one is an invalid ``proof'' for {\small $A \vee B \vdash A \wedge B$}. The others are correct derivations for {\small $A \vee A \vdash A$} and {\small $A \vdash A \wedge A$} (contraction and expansion edges are dotted).

\begin{figure}
\centering
\includegraphics[width=180pt]{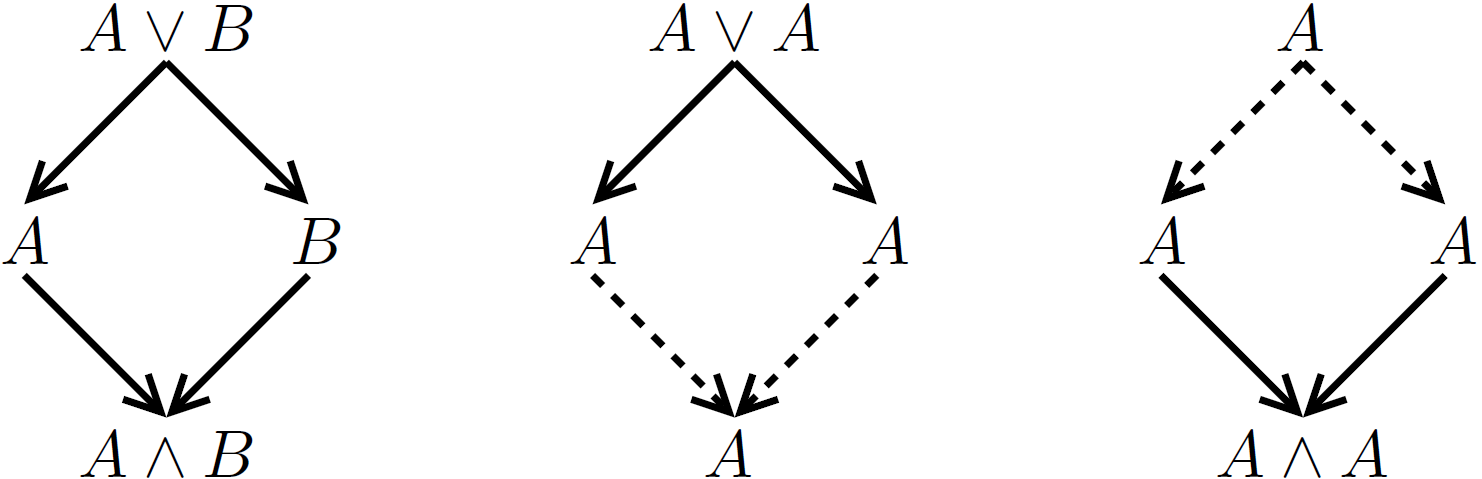}
\caption[Cycles.]{\centering Proof-graphs with cycles.}
\label{fi:cycle-examples}
\end{figure}

\subsubsection{Meta-edge and the scope of the hypothesis}

Besides expansion and contraction links there is the {\small $\rightarrow-I$} link. Both Ungar and Gentzen systems are formulated in such a way that
when the $\rightarrow$ connective is introduced, it may eliminate an arbitrary number of premises (including zero).
In N-Graphs this introduction is made in a more controlled way, which also complicates the task of
identifying inadequate proof-graphs.
For example, the first proof in Fig.~\ref{fi:metaExample} is not correct, but the second one is.

\begin{figure}
\centering
\includegraphics[width=280pt]{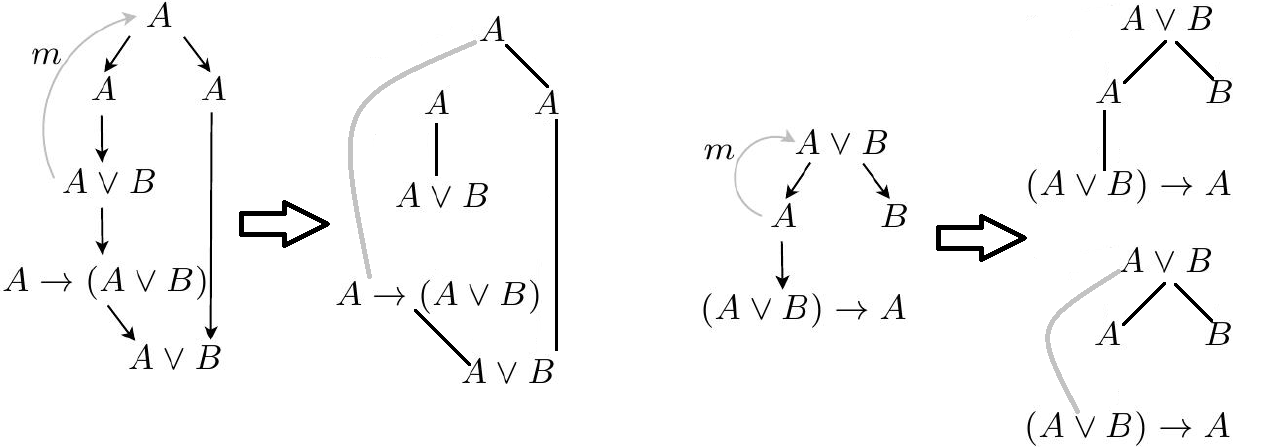}
\caption[Meta edge.]
{\centering Meta edge: an invalid application on the left for {\small $\vdash A \vee B$}
and a sound one on the right for {\small $\vdash (A \vee B) \rightarrow A, B$}.}
\label{fi:metaExample}
\end{figure}

\subsection{Soundness criteria}

Similar to Danos-Regnier criterion \cite{danos:1989}, we define the following subgraphs associated to a proof-graph.

\begin{Def}[Switching]
Given a proof-graph $G$, a \emph{switching graph} $S(G)$ associated with $G$ is a spanning subgraph\footnote{A spanning
subgraph is a subgraph $G_1$ of $G$ containing all the vertices of $G$.} of $G$ in which the following edges are removed:
one of the two edges of every expansion link and one of the two edges of every contraction link.
\end{Def}

\begin{Def}[Meta-switching, virtual edge]
Given a proof-graph $G$, a \emph{meta-switching graph} $S(G)$ associated with $G$ is a switching of $G$ in which every
link with meta-edge {\small $\{(u,w), (u,v)^m\}$} is replaced by one of the following edges: the one from {\small $u$}
to {\small $w$} or an  edge from {\small $v$} to {\small $w$}, which is defined as \emph{virtual edge}.

\end{Def}

\begin{Def}[N-Graph derivation]
A proof-graph $G$ is a \emph{N-Graph derivation} (or \emph{N-Graph} for short) iff  every meta-switching graph associated 
with $G$ is acyclic and connected.
\end{Def}

The focussing and defocussing links may also be classified according to their semantics.
The links {\small $\wedge-I$}, {\small $\bot-link$}, {\small $\rightarrow-E$},
{\small $\top-focussing~weak$} and \emph{expansion} are called \emph{conjunctive}.
The \emph{disjunctive} links are: {\small $\vee-E$}, {\small $\top-link$}, {\small $\rightarrow-I$},
{\small $\bot-defocussing~weak$} and \emph{contraction}. 
Here contraction and expansion draw attention: their geometry contradicts their semantic and they are switchable
(the ones that have one of its edge removed in every meta-switching).
Although focussing, the contraction has a disjunctive semantic; and the expansion is a conjunctive
link, even though defocussing. This means the formula occurences this links connect in a proof-graph
must be already connected some other way in order to the proof to be sound.
So the second and third proof-graphs in Fig. \ref{fi:cycle-examples} are N-Graphs, but the first one is not because its cycle is not valid.


The {\small $\rightarrow-I$} also plays an important role in the soundness criteria.
The premise of the link ({\small $B$}) and the canceled hypothesis ({\small $A$}) need to be already connected some other
way in the proof for it to be sound. Thus the meta-switching must choose to connect {\small $A \rightarrow B$} to {\small $A$} or
{\small $B$}.
In the first proof-graph of Fig. \ref{fi:metaExample} the conclusion of {\small $\rightarrow-I$} is
{\small $A \rightarrow (A \vee B)$}, so this formula already carries a dependency on {\small $A$} and the meta-edge
removes it from the set of premises. However, there is another occurrence of {\small $A$}, which is used by the
{\small $\rightarrow-E$} link to obtain a ``proof'' of {\small $\vdash A \vee B$}.

The soundness criteria captures this when the meta-switching choses the virtual edge, which links {\small $A \rightarrow (A \vee B)$} and {\small $A$}, and the result is not a tree. It does not occur with the other proof-graph of Fig.~\ref{fi:metaExample}: all the two meta-switchings are acyclic and connected.

Soundness and completeness of the system were proved through a mapping between N-Graphs and $LK$ (sequent calculus for classical logic)
\cite{oliveira:2001:phd,oliveira:2001} and in Section \ref{sec:sequent} we give a new proof of sequentialization.


\section{Sub-N-Graphs}

\begin{Def}[sub-N-Graph]
We say that {\small $H$} is a \emph{subproof-graph} of a proof-graph {\small $G$} if {\small $H$} is a subgraph of {\small $G$}
and {\small $H$} is a proof-graph.
If a vertex {\small $v$} labeled by a formula occurrence {\small $A$} is such that {\small $v \in PREMIS(H)$} 
({\small $v \in CONC(H)$}), then {\small $A$} is an \emph{upper} (\emph{lower}) \emph{door} of {\small $H$}.
If {\small $H$} is also a N-Graph, then it is a \emph{sub-N-Graph}.
\end{Def}

Let {\small $N_1$} and {\small $N_2$} be sub-N-Graphs of a N-Graph {\small $N$}.

\begin{Lem}[Union \cite{bellin:1995}]\label{lem:union}
{\small $N_1 \cup N_2$} is a N-Graph iff {\small $N_1 \cap N_2 \neq \emptyset$}.
\end{Lem}

\begin{proof}
Once {\small $N$} is a N-Graph, their meta-switchings do not have cycles and so any subgraph of {\small $N$}
may not have a cyclic meta-switching.
Then we must prove only the connectedness of all meta-switchings {\small $S$} associated with {\small $N_1 \cup N_2$}.
If {\small $N_1 \cap N_2 = \emptyset$}, then any meta-switching associated with {\small $N_1 \cup N_2$} is not connected.
Now let {\small $A \in N_1, ~B \in N_2, ~C \in N_1 \cap N_2$} and {\small $S(N_i)$} be the restriction of 
{\small $S(N_1 \cup N_2)$} to {\small $S(N_i)$}. Since {\small $N_1$} is a N-Graph,
there is a path between {\small $A$} and {\small $C$} in {\small $S(N_1)$}. For the same reason, there is a path between
{\small $B$} and {\small $C$} in {\small $S(N_2)$}. Thus there is a path between {\small $A$} and {\small $B$} in
{\small $S(N_1 \cap N_2)$} once {\small $C \in N_1 \cap N_2$}.
\qed
\end{proof}

\begin{Lem}[Intersection \cite{bellin:1995}]\label{lem:intersection}
If {\small $N_1 \cap N_2 \neq \emptyset$}, then {\small $N_1 \cap N_2$} is a N-Graph.
\end{Lem}

\begin{proof}
As in the previous lemma, it is sufficient to prove the connectivity of {\small $N_1 \cap N_2$}. 
Since {\small $N_1 \cap N_2 \neq \emptyset$}, let {\small $A \in N_1 \cap N_2$}.
If {\small $A$} is the only vertex present in {\small $N_1 \cap N_2$}, then it is connected and so is a N-Graph (axiom).
Otherwise, let {\small $B$} be any other vertex in {\small $N_1 \cap N_2$}, {\small $S$} be a meta-switching of 
{\small $N_1 \cap N_2$} and {\small $S_i$} be an extension of {\small $S$} for {\small $N_i$}. Once {\small $N_1$} and
{\small $N_2$} are sub-N-Graphs, there are a path {\small $\pi_1$} between {\small $A$} and {\small $B$} in 
{\small $S(N_1)$} and a path {\small $\pi_2$} in {\small $S(N_2)$}. If {\small $\pi_1 \neq \pi_2$}, then 
{\small $S_{12} = S_1 \cup S_2$} is a meta-switching for {\small $N_1 \cup N_2$} and {\small $S_{12}(N_1 \cup N_2)$}
has a cycle. So {\small $\pi_1 = \pi_2$} and {\small $A$} and {\small $B$} are connected.
\qed
\end{proof}

\begin{Def}[North, south and whole empires]
Let {\small $A$} be a formula occurrence in a N-Graph {\small $N$}. The \emph{north} (\emph{south}) 
\emph{empire of {\small $A$}}, represented by {\small $eA^{\wedge}$} ({\small $eA_{\vee}$}) 
is the largest sub-N-Graph of {\small $N$} having {\small $A$} as a \emph{lower} (\emph{upper}) door.
The \emph{whole empire of {\small $A$}} ({\small $wA$}) is the union of {\small $eA^{\wedge}$} and {\small $eA_{\vee}$}.
\end{Def}

If we prove that {\small $eA^{\wedge}$} and {\small $eA_{\vee}$} exist, then it is immediate the existence of {\small $wA$}
by lemma \ref{lem:union}. In the following section we give two equivalent constructions of empires and prove some properties.


\section{North and south empires}\label{sec:empires}

\subsection{Constructions and existence}\label{sub:empires_construction}

\begin{Def}[{\small $S^{\wedge}(N,A)$} and {\small $S_{\vee}(N,A)$}]
Let {\small $A$} be a formula occurrence in a N-Graph {\small $N$} and {\small $S$} an associated meta-switching
of {\small $N$}. If {\small $A$} is a premise of a link with a conclusion {\small $A'$} and the edge
{\small $(A,A')$} belongs to {\small $S(N)$}, then remove this edge and {\small $S^{\wedge}(N,A)$} is
the component that contains {\small $A$} and {\small $\overline{S^{\wedge}(N,A)}$} is the other one
(if {\small $A$} is premise of a disjunctive defocussing link different from {\small $\rightarrow-I$}, then
{\small $\overline{S^{\wedge}(N,A)}$} has two components). 
If {\small $A$} is not premise of any link in {\small $S(N)$}, then 
{\small $S^{\wedge}(N,A)$} is {\small $S(N)$} ({\small $\overline{S^{\wedge}(N,A)}$} is empty).
{\small $S_{\vee}(N,A)$} is defined analogously: if {\small $A$} is a conclusion of a link with a premise {\small $A''$} 
and the edge {\small $(A'',A)$} belongs to {\small $S(N)$}, then remove it and {\small $S_{\vee}(N,A)$} is the component
which has {\small $A$} and {\small $\overline{S_{\vee}(N,A)}$} is the other one (if {\small $A$} is conclusion of a
conjunctive focussing link, then {\small $\overline{S_{\vee}(N,A)}$} has two components). 
If {\small $A$} is not conclusion of any link in {\small $S(N)$}, then {\small $S_{\vee}(N,A)$} is equal to {\small $S(N)$}
({\small $\overline{S_{\vee}(N,A)}$} is empty).
\end{Def}

As the virtual edge added in some meta-switchings connects two conclusions of the {\small $\rightarrow-I$} link, we consider the main conclusion of the link {\small $(A \rightarrow B)$} as the conclusion of this virtual edge for the purpose of deciding if we must remove the edge to construct {\small $S^{\wedge}(N,A)$} or {\small $S_{\vee}(N,A)$}.

\begin{Def}[Principal meta-switching \cite{girard:1987,girard:1991}] \label{def:principal_switching}
Let {\small $A$} be a formula occurrence. We say that a meta-switching {\small $S_p^{\wedge}$}
({\small $S^p_{\vee}$}) is principal for {\small $eA^{\wedge}$} ({\small $eA_{\vee}$})
when it chooses the edges satisfying the following restrictions:
\begin{enumerate}
\item \label{item:principal_A}
\emph{{\small $\frac{A_{p_1} ~ A_{p_2}}{A_c}$} is a contraction link and a premise {\small $A_{p_i}$} is the 
formula occurrence {\small $A$}
({\small $\frac{A_p}{A_{c_1} ~ A_{c_2}}$} is an expansion link and a conclusion {\small $A_{c_i}$} is the formula occurrence 
{\small $A$})}: the meta-switching chooses the edge with {\small $A$}.
\item \label{item:principal_contraction}
\emph{{\small $\frac{X_{p_1} ~ X_{p_2}}{X_c}$} is a contraction link and only one premise {\small $X_{p_i} \neq A$} belongs to
{\small $eA^{\wedge}$} ({\small $eA_{\vee}$})}: the meta-switching links the conclusion with the premise which is not in
{\small $eA^{\wedge}$} ({\small $eA_{\vee}$}).
\item \label{item:principal_expansion}
\emph{{\small $\frac{X_p}{X_{c_1} ~ X_{c_2}}$} is an expansion link and only one conclusion {\small $X_{c_i} \neq A$}
belongs to {\small $eA^{\wedge}$} ({\small $eA_{\vee}$})}: {\small $S_p^{\wedge}$} ({\small $S^p_{\vee}$})
selects the edge which has the conclusion that is not in {\small $eA^{\wedge}$} ({\small $eA_{\vee}$}).
\item \label{item:principal_meta}
\emph{{\small $\frac{Y}{X ~ X \rightarrow Y}$} is a {\small $\rightarrow-I$} link and only {\small $X$} or
only {\small $Y$} belongs to {\small $eA^{\wedge}$} ({\small $eA_{\vee}$})}: {\small $S_p^{\wedge}$} ({\small $S^p_{\vee}$})
selects the edge which connects {\small $X \rightarrow Y$} to the other formula that is not in {\small $eA^{\wedge}$} ({\small $eA_{\vee}$}).
\item \label{item:principal_meta_A}
\emph{{\small $\frac{Y}{X ~ X \rightarrow Y}$}} is a {\small $\rightarrow-I$} link and {\small $X$} or {\small $Y$} is the
formula occurrence {\small $A$}: the meta-switching chooses the edge with {\small $A$}.
\end{enumerate}
\end{Def}

\begin{Lem} \label{lem:empires_constructions}
The north (south) empire of {\small $A$} exists and is given by the two following equivalent conditions:
\begin{enumerate}
\item \label{item:empire_switch}
{\small $\bigcap_{S}S^{\wedge}(N,A)$} ({\small $\bigcap_{S}S_{\vee}(N,A)$}), 
where {\small $S$} ranges over all meta-switchings of {\small $N$};

\item \label{item:empire_link}
the smallest set of formula occurrences of {\small $N$} closed under the following conditions:
\addtolength\leftmargini{0.5em}
\begin{enumerate}

\item \label{item:empires_base}
{\small $A \in eA^{\wedge}$} ({\small $A \in eA_{\vee}$});

\item \label{item:empires_up_simple}
if {\small $\frac{X}{Y}$} is a simple link and {\small $Y \in eA^{\wedge}$}, then {\small $X \in eA^{\wedge}$}
(if {\small $Y \neq A$} and {\small $Y \in eA_{\vee}$}, then {\small $X \in eA_{\vee}$});

\item \label{item:empires_up_conjunctive}
if {\small $\frac{X~Y}{Z}$} is a conjunctive focussing link and {\small $Z \in eA^{\wedge}$}, 
then {\small $X,Y \in eA^{\wedge}$}
(if {\small $Z \neq A$} and {\small $Z \in eA_{\vee}$}, then {\small $X,Y \in eA_{\vee}$});

\item \label{item:empires_up_disjunctive}
if {\small$\frac{~X}{Y~Z}$} is a disjunctive defocussing link different from {\small $\rightarrow-I$} and
{\small $Y \in eA^{\wedge}$} or {\small $Z \in eA^{\wedge}$}, then {\small $X \in eA^{\wedge}$}
(if {\small $Y \neq A \neq Z$} and {\small $Y \in eA_{\vee}$} or {\small $Z \in eA_{\vee}$}, then {\small $X \in eA_{\vee}$});

\item \label{item:empires_up_expansion}
if {\small $\frac{X_p}{X_{c_1} ~ X_{c_2}}$} is an expansion link and {\small $X_{c_1},X_{c_2} \in eA^{\wedge}$}, 
then {\small $X_p \in eA^{\wedge}$}
(if {\small $X_{c_1} \neq A \neq X_{c_2}$} and {\small $X_{c_1},X_{c_2} \in eA_{\vee}$}, then {\small $X_p \in eA_{\vee}$});

\item \label{item:empires_up_contraction}
if {\small $\frac{X_{p_1} ~ X_{p_2}}{X_c}$} is a contraction link and {\small $X_c \in eA^{\wedge}$}, then
{\small $X_{p_1},X_{p_2} \in eA^{\wedge}$}
(if {\small $X_c \neq A$} and {\small $X_c \in eA_{\vee}$}, then {\small $X_{p_1},X_{p_2} \in eA_{\vee}$});

\item \label{item:empires_up_meta}
if {\small $\frac{Y}{X ~ X \rightarrow Y}$} is a {\small $\rightarrow-I$} link and {\small $X \rightarrow Y \in eA^{\wedge}$}, 
then {\small $Y, X \in eA^{\wedge}$}
(if {\small $X \rightarrow Y \neq A$} and {\small $X \rightarrow Y \in eA_{\vee}$}, then {\small $Y, X \in eA_{\vee}$});

\item \label{item:empires_down_simple}
if {\small $\frac{X}{Y}$} is a simple link, {\small $X \neq A$} and {\small $X \in eA^{\wedge}$}, 
then {\small $Y \in eA^{\wedge}$}
(if {\small $X \in eA_{\vee}$}, then {\small $Y \in eA_{\vee}$});

\item \label{item:empires_down_conjunctive}
if {\small $\frac{X~Y}{Z}$} is a conjunctive focussing link, {\small $X \neq A \neq Y$} and {\small $X \in eA^{\wedge}$} 
or {\small $Y \in eA^{\wedge}$}, then {\small $Z \in eA^{\wedge}$}
(if {\small $X \in eA_{\vee}$} or {\small $Y \in eA_{\vee}$}, then {\small $Z \in eA_{\vee}$});

\item \label{item:empires_down_disjunctive}
if {\small $\frac{~X}{Y~Z}$} is a disjunctive defocussing link different from {\small $\rightarrow-I$},
{\small $X \neq A$} and {\small $X \in eA^{\wedge}$}, 
then {\small $Y,Z \in eA^{\wedge}$}
(if {\small $X \in eA_{\vee}$}, then {\small $Y,Z \in eA_{\vee}$});

\item \label{item:empires_down_expansion}
if {\small $\frac{X_p}{X_{c_1} ~ X_{c_2}}$} is an expansion link, {\small $X_p \neq A$} and {\small $X_p \in eA^{\wedge}$}, 
then {\small $X_{c_1},X_{c_2} \in eA^{\wedge}$}
(if {\small $X_p \in eA_{\vee}$}, then {\small $X_{c_1},X_{c_2} \in eA_{\vee}$});

\item \label{item:empires_down_contraction}
if {\small $\frac{X_{p_1} ~ X_{p_2}}{X_c}$} is a contraction link, {\small $X_{p_1} \neq A \neq X_{p_2}$} and 
{\small $X_{p_1},X_{p_2} \in eA^{\wedge}$}, then {\small $X_c \in eA^{\wedge}$}
(if {\small $X_{p_1},X_{p_2} \in eA_{\vee}$}, then {\small $X_c \in eA_{\vee}$});

\item \label{item:empires_down_meta}
if {\small $\frac{Y}{X ~ X \rightarrow Y}$} is a {\small $\rightarrow-I$} link, {\small $Y \neq A \neq X$} and {\small $X, Y \in eA^{\wedge}$}, 
then {\small $X \rightarrow Y \in eA^{\wedge}$}
(if {\small $X, Y \in eA_{\vee}$}, then {\small $X \rightarrow Y \in eA_{\vee}$}).

\end{enumerate}

\end{enumerate}

\end{Lem}

\begin{proof}
We will prove the case for {\small $eA^\wedge$} according to \cite{bellin:1995} (the case for {\small $eA_{\vee}$} is similar)

\begin{enumerate}[I]
\item \ref{item:empire_link} {\small $ \subseteq $} \ref{item:empire_switch}:
we show that \ref{item:empire_switch} is closed under conditions defining \ref{item:empire_link}. Its immediate
{\small $A \in \bigcap_{S}S^{\wedge}(N,A)$}
({\small $S^{\wedge}(N,A)$} contains {\small $A$} for every meta-switching {\small $S$}).
If {\small $B_1 \in \bigcap_{S}S^{\wedge}(N,A)$} and in all meta-switchings there is an edge {\small $(B_1,B_2)$}, then we 
conclude {\small $B_2 \in \bigcap_{S}S^{\wedge}(N,A)$} (imperialistic lemma \cite{girard:1991}). So the construction is also
closed under conditions \ref{item:empires_up_simple}, \ref{item:empires_up_conjunctive}, \ref{item:empires_up_disjunctive},
\ref{item:empires_down_simple}, \ref{item:empires_down_conjunctive} and \ref{item:empires_down_disjunctive}. 
Conditions \ref{item:empires_up_expansion}, \ref{item:empires_down_contraction} and \ref{item:empires_down_meta}
are also simple.
Now suppose that \ref{item:empire_switch} does not respect \ref{item:empires_up_contraction}.
Then there is a contraction link {\small $\frac{X_{p_1} ~ X_{p_2}}{X_c}$} such that 
{\small $X_c \in \bigcap_{S}S^{\wedge}(N,A)$}, but {\small $X_{p_i} \not\in \bigcap_{S}S^{\wedge}(N,A)$}, 
for {\small $i=1$} or {\small $i=2$}. Consider the first one: {\small $X_{p_1} \not\in S^{\wedge}(N,A)$} for some 
{\small $S$}. Since {\small $X_c \in S^{\wedge}(N,A)$}, then {\small $(X_{p_2},X_p) \in S(N,A)$} and so 
{\small $X_{p_2} \in S^{\wedge}(N,A)$}. Once {\small $\overline{S^{\wedge}(N,A)}$} is not empty, {\small $A$} must be 
premise of a link whose one conclusion is {\small $A'$} and {\small $A' \in \overline{S^{\wedge}(N,A)}$}. Let 
{\small $\pi$} be the path between {\small $X_{p_1}$} and {\small $A'$} in {\small $\overline{S^{\wedge}(N,A)}$}. Since
{\small $(X_{p_1}, X_c) \not\in S(N)$}, this edge does not belong to {\small $\pi$} (Fig. \ref{fi:empireEmpireComponent}).
Consider now a switch {\small $S'$} like
{\small $S$}, except that {\small $(X_{p_1}, X_c) \in S'(N)$}. Note that {\small $\pi$} is in {\small $S'(N)$} too and
{\small $X_c \in S'^{\wedge}(N,A)$} (because {\small $X_c \in \bigcap_{S}S^{\wedge}(N,A)$}). Then we may extend {\small $\pi$}
and get a path between {\small $A'$} and {\small $A$} without the edge {\small $(A,A')$} in {\small $S'(N)$}: we obtain
a cycle in {\small $S'(N)$}, which is a contradiction. Therefore \ref{item:empire_switch} is closed under
\ref{item:empires_up_contraction}. For similar reason, we conclude that \ref{item:empire_switch} is closed under 
\ref{item:empires_down_expansion} and \ref{item:empires_up_meta} too.

\begin{figure}[ht]
\centering
\includegraphics[width=150pt]{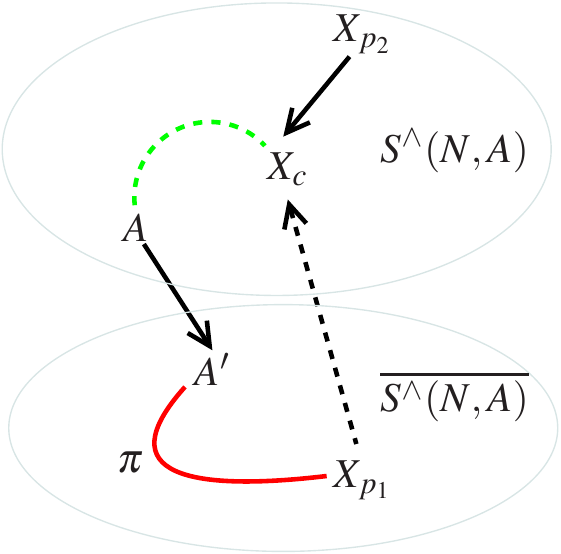}
\caption{If we choose the edge {\small $(X_{p_1},X_c)$}, we get a cycle.}
\label{fi:empireEmpireComponent}
\end{figure}

\item \ref{item:empire_switch} {\small $ \subseteq $} \ref{item:empire_link}: let {\small $S_p^{\wedge}$} a principal
meta-switching for {\small $eA^{\wedge}$}. We will prove {\small $S_p^{\wedge}(N,A) \subseteq $} \ref{item:empire_link}.
{\small $S_p^{\wedge}(N,A) \cap $} \ref{item:empire_link} {\small $\neq \emptyset$}, because both contain {\small $A$}.
But definition \ref{def:principal_switching} ensures that it is impossible to leave {\small $eA^{\wedge}$}
once we are in {\small $S_p^{\wedge}(N,A)$}.
Since {\small $\bigcap_{S}S^{\wedge}(N,A) \subseteq S_p^{\wedge}(N,A)$}, we conclude that
{\small $\bigcap_{S}S^{\wedge}(N,A) \subseteq eA^{\wedge}$}.

\end{enumerate}
\qed

\end{proof}

\begin{Cor}
{\small $S_p^{\wedge} = eA^{\wedge}$} and {\small $S^p_{\vee} = eA_{\vee}$}.
\end{Cor}

\begin{Cor}\label{cor:terminal_empire}
Let {\small $A$} be a premise\footnote{Note that it is not valid if {\small $A$} is a canceled hypothesis.} and {\small $B$} a conclusion. Then {\small $eA_{\vee} = eB^{\wedge} = N$}.
\end{Cor}

\begin{Lem}\label{lem:largest_empires}
{\small $eA^{\wedge}$} and {\small $eA_{\vee}$} are the largest sub-N-Graphs which contains {\small $A$} as a lower and
upper door, respectively.
\end{Lem}

\begin{proof}
The proof uses the same argument presented in \cite{bellin:1995} (see Proposition 2).
\qed
\end{proof}

Fig. \ref{fi:empire} illustrates some concepts about empires. For example, in the N-Graph on left, we have
{\small $eA^\wedge = \{A, ~A \vee C, ~C, ~C \wedge Z, ~Z, ~\neg{A} \wedge Z \}$} (formulas in green),
and {\small $eA_\vee = \{ A, ~\bot, ~\neg{A}, ~\neg{A} \wedge Z  \}$} (formulas in yellow).
The formula occurrence in red belongs to both empires.
We can see that there is no sub-N-Graph
which contains {\small $A$} as conclusion (premise) and is larger than {\small $eA^\wedge$} ({\small $eA_\vee$}), as
both conclusions of a expansion link are needed to add its premise (condition \ref{item:empires_up_expansion}).
For the second N-Graph we have the same color scheme for {\small $A$}, and here we can not have
the conclusion of the contraction link because we need both premises (condition \ref{item:empires_down_contraction}).

\begin{figure}[ht]
\centering
\includegraphics[width=300pt]{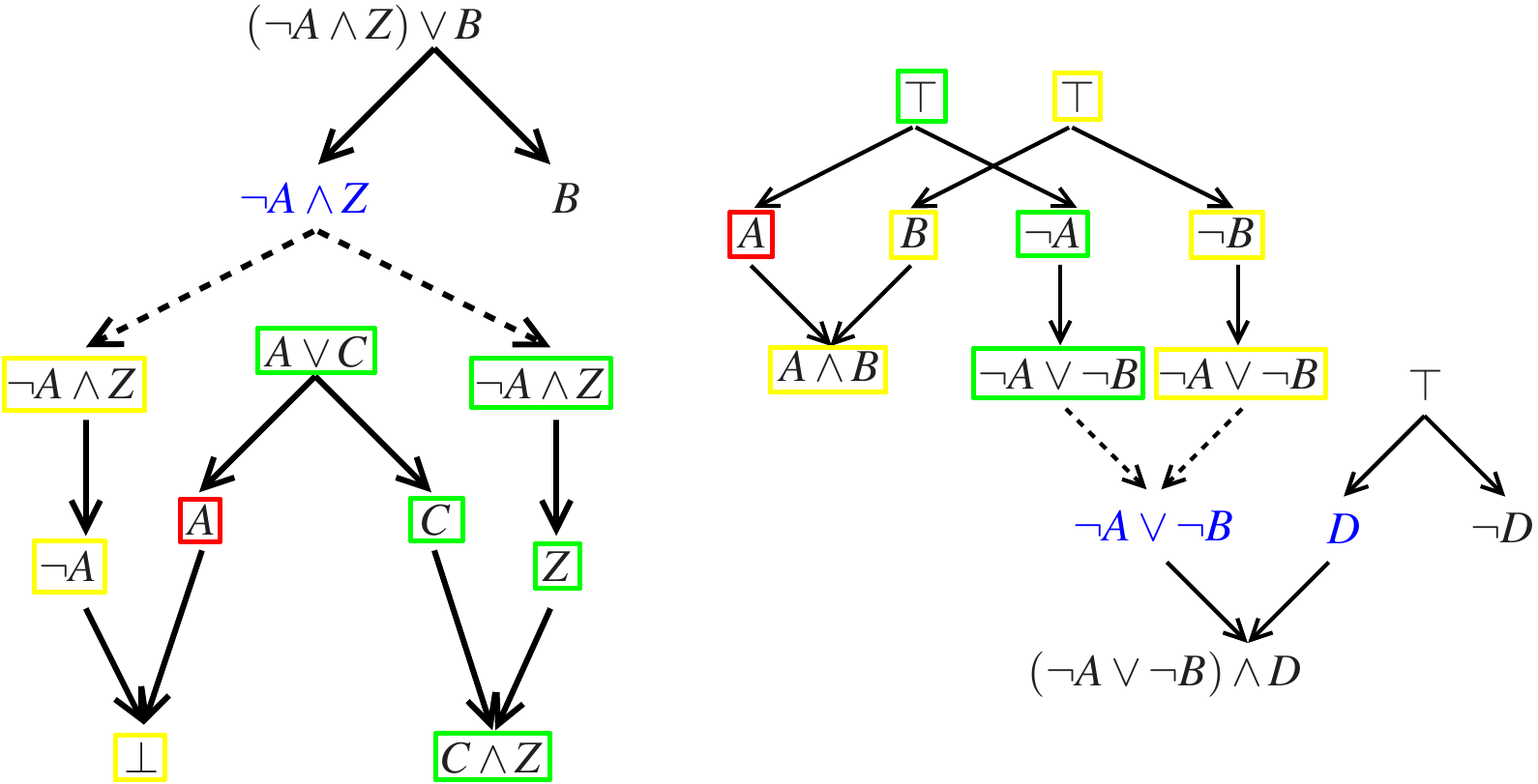}
\caption{N-Graphs for {\small $(\neg{A} \wedge Z) \vee B, ~A \vee C \vdash B, ~C \wedge Z$}
and {\small $ \vdash A \wedge B, ~(\neg{A} \vee \neg{B}) \wedge D, ~\neg{D}$}.}
\label{fi:empire}
\end{figure}

\subsection{Nesting lemmas}\label{sub:empires_nesting}


\begin{Lem}[Nesting of empires I \cite{girard:1991}]\label{lem:nesting_north_contain}
Let {\small $A$} and {\small $B$} be distinct formula occurrences in a N-Graph.
If {\small $A \in eB^{\wedge}$} and {\small $B \not\in eA^{\wedge}$}, then {\small $eA^{\wedge} \subsetneq eB^{\wedge}$}.
\end{Lem}

\begin{Lem}[Nesting of empires II \cite{girard:1991}]\label{lem:nesting_north_not_contain}
Let {\small $A$} and {\small $B$} be distinct formula occurrences in a N-Graph.
If {\small $A \not\in eB^{\wedge}$} and {\small $B \not\in eA^{\wedge}$}, 
then {\small $eA^{\wedge} \cap eB^{\wedge} = \emptyset$}.
\end{Lem}

\begin{proof}[Lemmas \ref{lem:nesting_north_contain} and \ref{lem:nesting_north_not_contain}]
Construct a principal meta-switching {\small $S_p^{\wedge}$} for {\small $eB^{\wedge}$} with some additional details:

\begin{enumerate}[I]
\item \label{item:nesting1_contraction}
\emph{contraction link whose conclusion belongs to {\small $eB^{\wedge}$}}: if the conclusion is not in 
{\small $eA^{\wedge}$}, then we proceed as we do for a principal meta-switching for {\small $eA^{\wedge}$} (if only one premise
is in {\small $eA^{\wedge}$}, choose the other premise);
\item \label{item:nesting1_expansion}
\emph{expansion link whose premise belongs to {\small $eB^{\wedge}$}}: if the premise is not in {\small $eA^{\wedge}$},
then we proceed as we do for a principal meta-switching for {\small $eA^{\wedge}$} (if only one conclusion is in 
{\small $eA^{\wedge}$}, choose the other conclusion);
\item \label{item:nesting1_meta}
\emph{{\small $\rightarrow-I$} link whose main conclusion belongs to {\small $eB^{\wedge}$}}: if the main conclusion is not in 
{\small $eA^{\wedge}$}, then we proceed as we do for a principal meta-switching for {\small $eA^{\wedge}$} (if only the premise or only the canceled hypothesis
is in {\small $eA^{\wedge}$}, choose the formula which is not in {\small $eA^{\wedge}$});
\item \label{item:nesting1_A}
\emph{if {\small $A$} is a premise of a link whose conclusion {\small $A'$} is in {\small $eB^{\wedge}$}}: then we choose
the edge {\small $(A,A')$}.
\end{enumerate}

First suppose {\small $A \in eB^{\wedge}$}.
We try to go from {\small $A$} to {\small $B$} without passing through {\small $(A,A')$}. Since {\small $S_p^{\wedge}$}
is principal for {\small $eB^{\wedge}$} and {\small $A \in eB^{\wedge}$}, all formulas in the path from {\small $A$} to
{\small $B$} belong to {\small $eB^{\wedge}$}. But {\small $B \not\in eA^{\wedge}$} and sometime we leave 
{\small $eA^{\wedge}$}. By construction \ref{item:empire_link} of lemma \ref{lem:empires_constructions}, there are only three
ways of leaving {\small $eA^{\wedge}$} without passing through {\small $(A,A')$}:
passing through a contraction link whose only one premise belongs to {\small $eA^{\wedge}$},
or passing through an expansion link whose only one conclusion belongs to {\small $eA^{\wedge}$},
or passing through a {\small $\rightarrow-I$} link whose only the premise or the canceled hypothesis is in
{\small $eA^{\wedge}$}; but, steps \ref{item:nesting1_contraction}, \ref{item:nesting1_expansion}
and \ref{item:nesting1_meta} avoid this cases, respectively.

Therefore it is impossible to leave {\small $eA^{\wedge}$} in {\small $S_p^{\wedge}(N,B)$}, unless 
{\small $(A,A') \in S_p^{\wedge}(N,B)$}. This implies {\small $S_p^\wedge(N,A) \subsetneq S_p^{\wedge}(N,B)$}.
Since {\small $eA^{\wedge} \subset S_p^\wedge(N,A)$} and {\small $eB^{\wedge} = S_p^{\wedge}(N,B)$}, we conclude
{\small $eA^{\wedge} \subsetneq eB^{\wedge}$}.

Now suppose {\small $A \not\in eB^{\wedge}$}.
\ref{item:nesting1_contraction}, \ref{item:nesting1_expansion} and \ref{item:nesting1_meta}
ensure we do not have any edges between {\small $eA^{\wedge}$} and 
{\small $\overline{eA^{\wedge}}$}\footnote{{\small $\overline{eA^{\wedge}}$} represents the set of all formula occurrences
which are not in {\small $eA^{\wedge}$}} in {\small $eB^{\wedge}$}, except perhaps for {\small $(A,A')$}.
But now {\small $A \not\in eB^{\wedge}$} and therefore
{\small $A \not\in S_p^{\wedge}(N,B)$}. So {\small $(A,A') \not\in S_p^{\wedge}(N,B)$}. Since 
{\small $eB^{\wedge} = S_p^{\wedge}$} and {\small $B \not\in eA^{\wedge}$}, no formula of {\small $eA^{\wedge}$} belongs
to {\small $eB^{\wedge}$} and thus {\small $eA^{\wedge} \cap eB^{\wedge} = \emptyset$}.
\qed
\end{proof}

From these two previous lemmas we have nesting lemmas \ref{lem:nesting_south_contain} and
 \ref{lem:nesting_south_not_contain} for south empires too (the proofs are similar to the previous ones) and from
these four nesting lemmas, it is possible to proof nesting lemmas between north and south
(\ref{lem:nesting_north_contain_south}, \ref{lem:nesting_south_contain_north} and \ref{lem:nesting_north_not_south}).

\begin{Lem}[Nesting of empires III \cite{girard:1991}]\label{lem:nesting_south_contain}
Let {\small $A$} and {\small $B$} be distinct formula occurrences in a N-Graph.
If {\small $A \in eB_{\vee}$} and {\small $B \not\in eA_{\vee}$},
then {\small $eA_{\vee} \subsetneq eB_{\vee}$}.
\end{Lem}

\begin{Lem}[Nesting of empires IV \cite{girard:1991}]\label{lem:nesting_south_not_contain}
Let {\small $A$} and {\small $B$} be distinct formula occurrences in a N-Graph.
If {\small $A \not\in eB_{\vee}$} and {\small $B \not\in eA_{\vee}$}, then {\small $eA_{\vee} \cap eB_{\vee} = \emptyset$}.
\end{Lem}

\begin{Lem}[Nesting of empires V]\label{lem:nesting_north_contain_south}
Let {\small $A$} and {\small $B$} be distinct formula occurrences in a N-Graph.
If {\small $A \in eB^{\wedge}$} and {\small $B \not\in eA_{\vee}$}, then {\small $eA_{\vee} \subsetneq eB^{\wedge}$}.
\end{Lem}

\begin{Lem}[Nesting of empires VI]\label{lem:nesting_south_contain_north}
Let {\small $A$} and {\small $B$} be distinct formula occurrences in a N-Graph. 
If {\small $A \in eB_{\vee}$} and {\small $B \not\in eA^{\wedge}$}, then {\small $eA^{\wedge} \subsetneq eB_{\vee}$}.
\end{Lem}

\begin{Lem}[Nesting of empires VII]\label{lem:nesting_north_not_south}
Let {\small $A$} and {\small $B$} be distinct formula occurrences in a N-Graph.
If {\small $A \not\in eB^{\wedge}$} and {\small $B \not\in eA_{\vee}$}, then {\small $eA_{\vee} \cap eB^{\wedge} = \emptyset$}.
\end{Lem}


\section{Whole empires}\label{sec:kingdoms}

We defined the \emph{whole empire of {\small $A$}} as the union of the north and the south empires of {\small $A$}.
Now we use the north and south empires properties to find new ones about whole empires.

\begin{Lem}\label{lem:subNgraph_kingdom}
{\small $wA$} is a sub-N-graph.
\end{Lem}

\begin{proof}
Once we proved that {\small $eA^{\wedge}$} and {\small $eA_{\vee}$} are N-graphs (lemma \ref{lem:largest_empires}) and
{\small $eA^{\wedge} \cap eA_{\vee} = \{A\}$}, we get {\small $wA = eA^{\wedge} \cup eA_{\vee}$} is a sub-N-graph 
by lemma \ref{lem:union}.
\qed
\end{proof}

\begin{Cor}\label{cor:terminal_kingdom}
Let {\small $A$} be a premise and {\small $B$} a conclusion. Then {\small $wA = wB = N$} 
(by corollary \ref{cor:terminal_empire}).
\end{Cor}

\begin{Lem}[Nesting of whole empires I]\label{lem:nesting_kingdom_not_contain}
Let {\small $A$} and {\small $B$} be distinct occurrences. If {\small $A \not\in wB$} and {\small $B \not\in wA$},
then {\small $wA \cap wB = \emptyset$}.
\end{Lem}

\begin{proof}
Since {\small $A \not\in wB$} and {\small $B \not\in wA$}, we get: {\small $A \not\in eB^{\wedge}$}, 
{\small $A \not\in eB_{\vee}$}, {\small $B \not\in eA^{\wedge}$} and {\small $B \not\in eA_{\vee}$}.
We apply nesting of empires lemmas:
\begin{enumerate}
\item \label{item:nesting_kingdom_A_north_B_north}
if {\small $B \not\in eA^{\wedge}$} and {\small $A \not\in eB^{\wedge}$}, then 
{\small $eA^{\wedge} \cap eB^{\wedge} = \emptyset$} (by lemma \ref{lem:nesting_north_not_contain});

\item \label{item:nesting_kingdom_A_south_B_north}
if {\small $B \not\in eA_{\vee}$} and {\small $A \not\in eB^{\wedge}$}, then 
{\small $eA_{\vee} \cap eB^{\wedge} = \emptyset$} (by lemma \ref{lem:nesting_north_not_south});

\item \label{item:nesting_kingdom_B_north_inter}
uniting \ref{item:nesting_kingdom_A_north_B_north} and \ref{item:nesting_kingdom_A_south_B_north} and applying the 
distributive law: {\small $eB^{\wedge} \cap (eA^{\wedge} \cup eA_{\vee}) = \emptyset$};

\item \label{item:nesting_kingdom_A_north_B_south}
if {\small $B \not\in eA^{\wedge}$} and {\small $A \not\in eB_{\vee}$}, then 
{\small $eA^{\wedge} \cap eB_{\vee} = \emptyset$} (by lemma \ref{lem:nesting_north_not_south});

\item \label{item:nesting_kingdom_A_south_B_south}
if {\small $B \not\in eA_{\vee}$} and {\small $A \not\in eB_{\vee}$}, then 
{\small $eA_{\vee} \cap eB_{\vee} = \emptyset$} (by lemma \ref{lem:nesting_south_not_contain});

\item \label{item:nesting_kingdom_B_south_inter}
uniting \ref{item:nesting_kingdom_A_north_B_south} and \ref{item:nesting_kingdom_A_south_B_south} and applying the 
distributive law: {\small $eB_{\vee} \cap (eA^{\wedge} \cup eA_{\vee}) = \emptyset$};

\item \label{item:nesting_kingdom_last}
uniting \ref{item:nesting_kingdom_B_north_inter} and \ref{item:nesting_kingdom_B_south_inter}, the 
distributive law: {\small $(eA^{\wedge} \cup eA_{\vee}) \cap (eB^{\wedge} \cup eB_{\vee}) = \emptyset$}.
\end{enumerate}
\qed
\end{proof}

\begin{Lem}[Nesting of whole empires II]\label{lem:nesting_kingdom_contain}
Let {\small $A$} and {\small $B$} be distinct occurrences. If {\small $A \in wB$} and {\small $B \not\in wA$},
then {\small $wA \subsetneq wB$}.
\end{Lem}

\begin{proof}
Once {\small $B \not\in wA$}, we have {\small $B \not\in eA^{\wedge}$} and {\small $B \not\in eA_{\vee}$}.
For {\small $A \in wB$} we get {\small $A \in eB^{\wedge}$} or {\small $A \in eB_{\vee}$}. We will prove the lemma
for {\small $A \in eB^{\wedge}$} (the case for south is analogous):

\begin{enumerate}
\item if {\small $A \in eB^{\wedge}$} and {\small $B \not\in eA^{\wedge}$}, then
{\small $eA^{\wedge} \subsetneq eB^{\wedge}$} (by lemma \ref{lem:nesting_north_contain});

\item if {\small $A \in eB^{\wedge}$} and {\small $B \not\in eA_{\vee}$}, then
{\small $eA_{\vee} \subsetneq eB^{\wedge}$} (by lemma \ref{lem:nesting_north_contain_south});

\item if {\small $eA^{\wedge} \subsetneq eB^{\wedge}$} and {\small $eA_{\vee} \subsetneq eB^{\wedge}$} , then
{\small $eA^{\wedge} \cup eA_{\vee} \subsetneq eB^{\wedge}$};

\item if {\small $eA^{\wedge} \cup eA_{\vee} \subsetneq eB^{\wedge}$}, then 
{\small $eA^{\wedge} \cup eA_{\vee} \subsetneq eB^{\wedge} \cup eB_{\vee}$}.
\end{enumerate}
\qed
\end{proof}

\begin{Def}[{\small $\ll$}]\label{def:order}
Let {\small $A$} and {\small $B$} be formula occurrences of {\small $N$}. We say {\small $A \ll B$} iff {\small $wA \subsetneq wB$}.
\end{Def}

It is immediate that {\small $\ll$} is a strict ordering of formula occurrences of {\small $N$} which are not premises
neither conclusions, since we have for any
domain set {\small $X$} and any subset {\small $Q$} of {\small $\wp(X)$}\footnote{{\small $\wp(X)$} is the power set of 
{\small $X$}}, {\small $(\subseteq, Q)$}
is a poset. Maximal formulas with regard to {\small $\ll$} will split {\small $N$}. 
Given that the whole empires of premises and conclusions are always equal to {\small $N$} by corollary
\ref{cor:terminal_kingdom}, we are not interested in these formulas. So they are not in the domain of {\small $\ll$}.
One may easily verify that if there are no contraction, extension and {\small $\rightarrow-I$}
links, for all formula {\small $A$} of {\small $N$}, {\small $wA=N$} and so any formula would be maximal.
The next three following lemmas show how these links act on {\small $\ll$}.

\begin{Lem}\label{lem:kingdom_order_contraction}
Let {\small $l=\frac{\ldots X \ldots}{\ldots Y \ldots}$} be a link different from {\small $\rightarrow-I$} such that there
is a formula-occurrence {\small $A$} which {\small $X \in wA$} and {\small $Y \not\in wA$}. Then {\small $A \ll Y$}.
\end{Lem}

\begin{proof}
Once {\small $X \in eA^{\wedge} \cup eA_{\vee}$}, we have two cases. If
\emph{{\small $X \in eA^{\wedge}$}}, then
since {\small $Y \not\in eA^{\wedge}$}, {\small $l$} must be a contraction link and its other premise does not belong to 
{\small $eA^{\wedge}$} (construction \ref{item:empire_link} in lemma \ref{lem:empires_constructions}). Therefore {\small $Y$}
is a conclusion of a contraction link and this implies {\small $X \in eY^{\wedge}$} (by \ref{item:empires_up_contraction} in
lemma \ref{lem:empires_constructions}). So {\small $eA^{\wedge} \cap eY^{\wedge} \neq \emptyset$}. 
If {\small $A \not\in eY^{\wedge}$}, then we will have {\small $eA^{\wedge} \cap eY^{\wedge} = \emptyset$} 
(by lemma \ref{lem:nesting_north_not_contain}): a contradiction. Thus {\small $A \in eY^{\wedge}$} and, by lemma 
\ref{lem:nesting_kingdom_contain}, we conclude {\small $wA \subsetneq wY$}.
The case for \emph{{\small $X \in eA_{\vee}$}} is analogous.

\qed
\end{proof}

Next lemma is similar, but for expansion link: 
\begin{Lem}\label{lem:kingdom_order_expansion}
Let {\small $l=\frac{\ldots X \ldots}{\ldots Y \ldots}$} be a link different from {\small $\rightarrow-I$} such that there
is a formula-occurrence {\small $A$} which {\small $X \not\in wA$} and {\small $Y \in wA$}. Then {\small $A \ll X$}.
\end{Lem}

The following lemma is the corresponding for {\small $\rightarrow-I$} link and its use same ideas as above.

\begin{Lem}\label{lem:kingdom_order_meta}
Let {\small $l=\frac{Y}{X ~ X \rightarrow Y}$} be a {\small $\rightarrow-I$} link such that there is a formula occurrence
{\small $A$} which {\small $Y \in wA$} or {\small $X \in wA$}, but {\small $(X \rightarrow Y) \not\in wA$}.
Then {\small $A \ll (X \rightarrow Y)$}.
\end{Lem}


\section{Sequentialization}\label{sec:sequent}

We saw in Sections \ref{sec:empires} and \ref{sec:kingdoms} how to define empires for proof-graphs with switchable defocussing
links (expansion) and proved some properties. Now we will show a new proof of sequentialization for these proof-graphs.
Without loss of generality, we assume {\small $\top$} as {\small $A \vee \neg{A}$} and {\small $\bot$} as 
{\small $A \wedge \neg{A}$}, where the formula {\small $A$} belongs to the premise or conclusion of the link.

\begin{Theo}[Sequentialization]
\label{theo:sequentialization_fragment}
Given a N-Graph derivation {\small $N$}, there is a sequent calculus derivation 
{\small $SC(N)$} of {\small $A_1,\dots,A_n \vdash B_1,\dots,B_m$} in the classical sequent calculus whose occurrences
of formulas {\small $A_1,\dots,A_n$} and {\small $B_1,\dots,B_m$} are in one-to-one
correspondence with the elements of {\small $PREMIS(N)$} and {\small $CONC(N)$}, respectively.
\end{Theo}

\begin{proof}
We proceed by induction on the number of links of {\small $N$}.

\begin{enumerate}
\item \label{item:sequent_base_case}
\emph{{\small $N$} does not have any link (it has only one vertex {\small $v$} labelled with {\small $A$})}:
this case is immediate. {\small $SC(N)$} is {\small $A \vdash A$}.

\item \label{item:sequent_one_link}
\emph{{\small $N$} has only one link}:
since {\small $N$} is a N-Graph, then this link is not a contraction, an expansion, neither a {\small $\rightarrow-I$}.
This case is simple, once there is a simple mapping between links and sequent calculus rules,
which makes the construction of {\small $SC(N)$} immediate (completeness proof \cite{oliveira:2001:phd,oliveira:2001}).
For example, in case {\small $\wedge-I$}:
$$\frac{A \vdash A \quad B \vdash B}{A, B \vdash A \wedge B} ~\wedge - R$$

\item \label{item:sequent_expansion}
\emph{{\small $N$} has an initial expansion link (the premise of the link is a premise of {\small $N$})}:
the induction hypothesis has built a derivation {\small $\Pi$} ending with 
{\small $\mathbf{A}, \mathbf{A}, \ldots, A_n \vdash B_1, \ldots, B_m$}. Then {\small $SC(N)$} is achieved by 
left contraction:
$$\Pi \quad \atop \displaystyle{\frac{\mathbf{A}, \mathbf{A}, \ldots, A_n \vdash B_1, \ldots, B_m}
{\mathbf{A}, \ldots, A_n \vdash B_1, \ldots, B_m}} ~ LC$$

\item \label{item:sequent_contraction}
\emph{{\small $N$} has a final contraction link (the conclusion of the link is a conclusion of {\small $N$})}:
here the induction hypothesis has built a derivation {\small $\Pi$} ending with 
{\small $A_1, \ldots, A_n \vdash \mathbf{B}, \mathbf{B}, \ldots, B_m$}. Hence {\small $SC(N)$} is obtained by
right contraction:
$$\Pi \qquad \quad \atop \displaystyle{\frac{A_1, \ldots, A_n \vdash \mathbf{B}, \mathbf{B}, \ldots, B_m}
{A_1, \ldots, A_n \vdash \mathbf{B}, \ldots, B_m}} ~RC$$

\item \label{item:sequent_meta}
\emph{{\small $N$} has a final {\small $\rightarrow-I$} link (the main conclusion of the link is a conclusion
of {\small $N$})}:
here the induction hypothesis has built a derivation {\small $\Pi$} ending with 
{\small $A_1, \ldots, A_n, \mathbf{A} \vdash \mathbf{B}, B_1, \ldots, B_m$}. Hence {\small $SC(N)$} is obtained by
{\small $\rightarrow-R$}:
$$\Pi \qquad \quad \quad \atop \displaystyle{\frac{A_1, \ldots, A_n, \mathbf{A} \vdash \mathbf{B}, B_1, \ldots, B_m}
{A_1, \ldots, A_n \vdash \mathbf{A \rightarrow B}, B_1, \ldots, B_m}} ~\rightarrow-R$$

\item \label{item:sequent_split}
\emph{{\small $N$} has more than one link, but no initial expansion link, no final contraction link and no final
{\small $\rightarrow-I$} link}:
this case is more complicated and is similar to that one in MLL$^-$ in which all terminal links are 
{\small $\otimes$}.
Yet here we have an additional challenge: the split node is in the middle of the proof.
Choose a formula occurrence {\small $A$} which is maximal with respect to {\small $\ll$}.
We claim that {\small $wA = eA^{\wedge} \cup eA_{\vee} = N$}. That is, {\small $A$} labels the split node.

Suppose not. Then let {\small $Z$} be a formula occurrence such that {\small $Z \in N - (eA^{\wedge} \cup eA_{\vee})$} 
and {\small $S_p^{\wedge}$} be a principal meta-switching for {\small $eA^{\wedge}$}. 
Given that {\small $Z \not\in eA^{\wedge}$}, the path
{\small $\rho$} from {\small $A$} to {\small $Z$} in {\small $S_p^{\wedge}(N)$} passes through a conclusion {\small $A'$} of
{\small $A$}. Let {\small $A_{\vee}$} be the last node which belongs to {\small $eA_{\vee}$} in {\small $\rho$} and 
{\small $W$} the next one in {\small $\rho$} (i.e. {\small $W \not\in eA_{\vee}$}). There are two cases for the edge
incident to {\small $A_{\vee}$} and {\small $W$}:

\begin{enumerate}
\item \emph{{\small $(A_{\vee},W)$} belongs to a contraction link whose other premise is not in {\small $eA_{\vee}$}}:
according to lemma \ref{lem:kingdom_order_contraction} we have {\small $A \ll W$}, contradicting the maximality of
{\small $A$} in {\small $\ll$}.
\item \emph{{\small $(W,A_{\vee})$} belongs to an expansion link whose other conclusion is not in {\small $eA_{\vee}$}}:
we apply lemma \ref{lem:kingdom_order_expansion} and conclude {\small $A \ll W$} here too. We contradict our choice again.
\item \emph{{\small $(A_{\vee},W)$} belongs to a {\small $\rightarrow-I$} whose {\small $W$} is the main conclusion and
the other formula is not in {\small $eA_{\vee}$}}: we use lemma \ref{lem:kingdom_order_meta} and also get
{\small $A \ll W$}.
\end{enumerate}

Thus {\small $wA = N$}. Let {\small $\Gamma_1, \Gamma_2, \Delta_1, \Delta_2$} be sets of formula occurrences such that:
{\small $\Gamma_1 \cup \Gamma_2 = \lbrace A_1, \ldots, A_n \rbrace$}, {\small $\Delta_1 \cup \Delta_2 = \lbrace B_1, \ldots, B_m \rbrace$} 
and {\small $\Gamma_1 \cap \Gamma_2 = \Delta_1 \cap \Delta_2 = \emptyset$}. 
Since {\small $eA^{\wedge}$} is a N-Graph and {\small $A$}
is a lower door, the induction hypothesis built {\small $SC(eA^{\wedge})$} ending with {\small $\Gamma_1 \vdash \Delta_1, A$}.
Once {\small $eA_{\vee}$} is a N-Graph and {\small $A$} is an upper door, the induction hypothesis made 
{\small $SC(eA_{\vee})$} ending with {\small $A, \Gamma_2 \vdash \Delta_2$}. So {\small $SC(N)$} is achieved by cut rule:
$$\displaystyle{
\displaystyle{SC(eA^{\wedge}) \atop \Gamma_1 \vdash \Delta_1, ~A} \qquad \displaystyle{SC(eA_{\vee}) \atop A, ~\Gamma_2 \vdash \Delta_2}
\over \Gamma_1, \Gamma_2 \vdash \Delta_1, \Delta_2}$$

\end{enumerate}
\qed
\end{proof}


\section{Conclusion}\label{sec:conclusion}

With N-Graphs, the structural links are based on the sequent calculus, but the logical links emulate the rules of ND.
Sequent calculus (classical and linear) have only introduction rules.
On the other hand, natural deduction and N-Graphs present elimination rules, so we need to adapt the notion of empire
from proof-nets to account for multiple-conclusion ND. This was done with south empires.

Another feature of N-Graphs and ND is the presence of improper rules. The introduction of ``{\small $\rightarrow$}''
has correctness criteria that do not apply only to local formulas, but to the whole deduction.
Thus, besides the presence of expansion link, which is  
defocussing and switchable, the sequentialization proof showed here also to accomplish a global rule in a system
that adopts Danos-Regnier's criteria.

By mapping derivations in sequent calculus to derivations in ND-style in the way we have just described we have been able to
formulate a new and rather general method of performing surgical cuts on proofs in multiple conclusion ND,
giving rise to subnets for classical logic.
In N-Graphs the split nodes, maximal with respect to the ordering induced by the union of the empires, may be
located anywhere in the proof, not only as a terminal node representing a conclusion.
We show an example in Fig. \ref{fi:empire}: every initial link is defocussing unswitchable and every final link is
focussing unswitchable in both N-Graphs. Their maximal nodes are highlighted and they split the proofs into two
correct proofs. This is illustrated in Fig. \ref{fi:cut} for the N-Graph on the left (Fig. \ref{fi:empire}).
The same procedure could also be applied to any of the two maximal nodes in the N-Graph on the right.

\begin{figure}
\centering
\includegraphics[width=300pt]{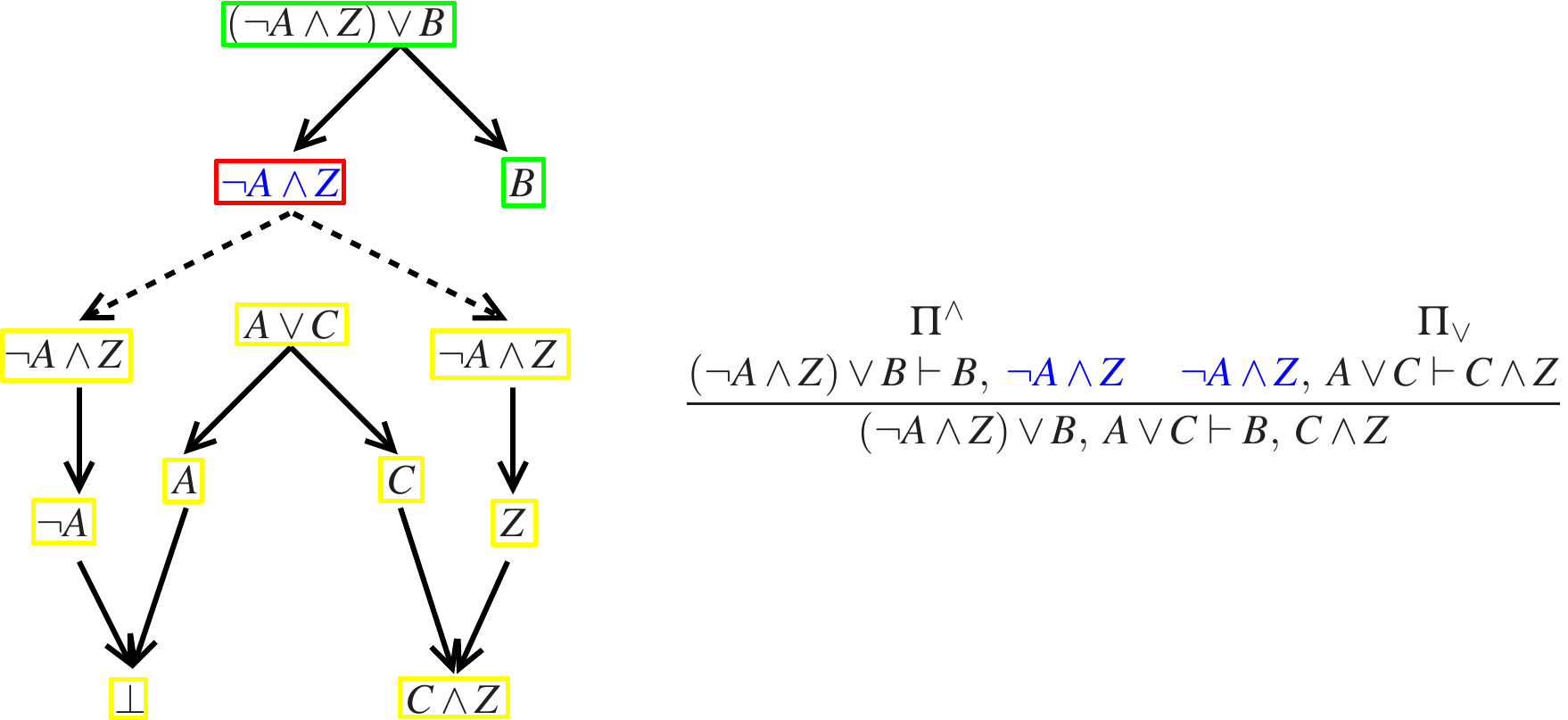}
\caption{Example of how to cut using the maximal node.}
\label{fi:cut}
\end{figure}



\end{document}